\documentclass[review]{elsarticle}
\usepackage{amsmath}
\usepackage{ntheorem}
\usepackage{supertabular}
\usepackage{graphicx}
\usepackage{epstopdf}
\newtheorem{theorem}{Theorem}
\usepackage{amssymb}
\newtheorem*{proof}{Proof}
\journal{Chaos, Solitons \& Fractals}
\usepackage{makecell}
\usepackage{booktabs} 
\usepackage{tikz}

\newtheorem{thm}{Theorem}
\newtheorem{lem}[thm]{Lemma}

\newtheorem{cor}{Corollary}
\newtheorem{definition}{Definition}
\theoremstyle{definition}

\theoremstyle{remark}

\begin{document}

\begin{frontmatter}

\title{General non-linear imitation leads to limit cycles in eco-evolutionary dynamics}
\author{Yuan Liu$^1$}
\author{Lixuan Cao$^1$}
\author{Bin Wu$^1$\corref{mycorrespondingauthor}}
\cortext[mycorrespondingauthor]{Corresponding author (bin.wu@bupt.edu.cn)}
\address{$^1$ School of Sciences, Beijing University of Posts and Telecommunications, 100876, China}

\begin{abstract}
Eco-evolutionary dynamics is crucial to understand how individuals' behaviors and the surrounding environment interplay with each other. Typically, it is assumed that individuals update their behaviors via linear imitation function, i.e., the replicator dynamics. It has been proved that there cannot be limit circles in such eco-evolutionary dynamics. It suggests that eco-evolutionary dynamics alone is not sufficient to explain the widely observed fluctuating behavior in both the biological and social systems. We extrapolate from the linear imitation function to general imitation function, which can be non-linear. It is shown that the general imitation does not change the internal equilibrium and its local stability. It however leads to limit cycles, which are never present in classical eco-evolutionary dynamics. Moreover, the average cooperation level and average environment state agree with the first and the second component of the internal fixed point. Furthermore, we estimate the location of the emergent limit cycle. On the one hand, our results provide an alternative mechanism to make cooperators and defectors coexist in a periodic way in the "tragedy of the commons" and indicates the global eco-evolutionary dynamics is sensitive to the imitation function. Our work indicates that the way of imitation, i.e., the imitation function, is crucial in eco-evolutionary dynamics, which is not true for evolutionary dynamics with a static environment.
\end{abstract}


\end{frontmatter}


\section{Introduction}
Social dilemmas occur if individual interests differ from collective welfare in a limited public good \cite{dawes1980social,doebeli2005models,hector2007biodiversity}.
Individuals tend to adopt strategies that maximize their benefits.
However, the group benefit will be maximal if they leverage the public good with restraint.
Garrett Hardin terms this dilemma as ``tragedy of the commons'' and claims that it is inevitable \cite{hardin1998extensions}.
It is of central importance in biology and socialogy to answer why and how cooperation evolves by natural selection \cite{Nowak_2010, Rainey_2003, West_2006, _gren_2019, doi:10.1126/science.309.5731.93}.
The Prisoner's Dilemma is a paradigm to study this issue. 
In Prisoner's Dilemma, cooperators help others at a cost but defectors obtain benefits without any costs. 
The dynamics are typically modeled by the replicator equation in a well-mixed and inﬁnite population \cite{Nowak_2004, hofbauer_sigmund_1998, Nowak2006, doi:10.1126/science.1093411, TAYLOR1978145}. 
Imitation rule is widely adopted, which implies that an individual imitates other's strategy with a probability proportional to the difference between their payoffs. 
Over the past few decades, various mechanisms have been proposed to promote cooperation: direct reciprocity, indirect reciprocity, kin selection, network reciprocity and group selection \cite{PMID:17158317, NOWAK1998561, nowak1998, nowak2005, taylor1996, nowak1992, May2006}. 

The environmental feedback is proposed as a new mechanism to explain the emergence and maintenance of cooperation in social dilemmas \cite{weitz2016oscillating, Szolnoki2018, CAO2021111088, LIU2022127309}, including ``tradegy of the commons", vaccination dilemma and climate change dilemma \cite{ZHAO2021125963,arefin2020,liu2022}.
Environmental feedback is universal in natural systems \cite{lee2019social,grman2012ecological,stewart2014collapse,cortez2018destabilizing,grunert2021evolutionarily}.
For example, some bacteria produce siderophore to remove iron from the host, while some bacteria do not produce siderophores, but absorb siderophores after the combination of them. 
This kind of bacteria spread fast among siderophore producers when the content of siderophore is high, which will lead to the reduction of siderophore content after a period of time \cite{west2003cooperation,ross2015evolutionary,kamran2019}.
Besides, the epidemic dies out if individuals follow disease intervention measures (such as vaccination, social distancing, mask wearing, testing and isolation).
However, individuals become less vigilant and choose not to comply with public measures if epidemic is not serious \cite{Glaubitz_2020,Stanley2021,bauch2003group,bauch2004vaccination,galvani2007long}.
In other words, it is assumed that cooperators enhance the environment, replete environment breeds defection, then
defectors destroy the environment and depleted environment is conductive to evolution of cooperators in evolutionary game with feedback mechanism \cite{greig2004prisoner,gore2009snowdrift,menge2008evolutionary,beekman2001evolution}.

Many studies focus on the feedback between human behavior and environment \cite{weitz2016oscillating, Szolnoki2018, CAO2021111088, LIU2022127309, gong2020limit, liu2022, farahbakhsh2022modelling, Sanchez2013}. 
Weitz et al. has proposed the replicator dynamics of evolutionary game with feedback mechanism. 
They described the dilemma ``tradegy of the commons'' mathematically and draw a conclusion that ``tragedy of the commons'' can be averted, that is, cooperation appears \cite{weitz2016oscillating}.
Then the emergence of group cooperation in public goods game based on environmental feedback has been proposed \cite{wang2020eco}.
In the Weitz's model, individuals update their behavior via replicator equation, which results from the linear imitation function of payoff differences. 
However, the model is not sufficient to explain eco-evolutionary systems in which the internal period orbits are present. 
The phenomena of oscillatory coexistence are widely found in evolutionary dynamics. 
For example, in Lotka-Volterra dynamics, there are two species, one is a predator and the other is prey. 
It is assumed that the prey has an unlimited food supply and to reproduce exponentially, while the predator declines exponentially due to natural death or migration. 
The rate of predation upon the prey is assumed to be proportional to the rate at which the predators and the prey meet. 
It leads to oscillatory periodic solutions. 
Another example is vaccination in public health: the vaccinators pay costs to gain immunity. 
The herd immunity protects not only themselves but also the unvaccinators, who do not pay any cost. 
The unvaccinators, however, can pay more cost for recovery if they are infected. 
On the other hand, more vaccinators decrease the perceived vaccination risk since individuals would take it as safe if more are vaccinated. 
And low vaccination risk improves the vaccination level. 
In this case, the vaccinators and the unvaccinators can lead to oscillatory coexistence \cite{liu2022}. 
The oscillatory dynamics is also present in biological systems: A yeast secrets the enzymes to digest sucrose outside of the cell. 
Those who do not secrets exploit the enzyme produced by a neighbor cell. 
The two types of yeast oscillatory coexist in an everchanging environment. 
``The tragedy of the commons” suggests that the competing behavior of parasites that through acting selfishly eventually destroy their common host \cite{greig2004prisoner, doebeli2004evolutionary}. 
However, the study shows that within-host competition detracts from the ability of viruses to exploit the host when coinfection occurs \cite{turner1998sex, brown2002does, killingback2010diversity}. 
Besides, the phenomena including moderate fishing, water conservation and conversion of farmland to forests also reflect the existence of limit cycles. 
Therefore, the oscillatory coexistence is widely present ranging from biology to social systems. 
However, there is a gap between these oscillatory coexistence phenomena and the Weitz’s model, in which there cannot be any oscillatory coexistence. 
It indicates that additional mechanisms are required to explain these fluctuating behaviors in biological and social systems.

In fact, many studies have shown that imitation is a profound social learning process \cite{over2020, farmer2018} and individual's imitation behavior seems to be traceable to other simple mechanisms \cite{APESTEGUIA2007217,PhysRevE.94.012124,DUERSCH201288}.
Fermi update rule and replicator equation are usually adopted in evolutionary dynamics \cite{nowak2006evolutionary,zeeman1980population,nowak2004evolutionary, chowdhury2021eco, roy2022eco}.
Individuals tend to imitate others' more successful behavior \cite{bauch2005, Farahbakhsh2021}, such as in social systems \cite{masuda2014,over2020} and vaccination behavior \cite{fu2011,wu2010universality,traulsen2005}.
It is assumed that individuals imitate the opponent's strategy according to the probability of the difference in their payoffs as well as the imitation intensity that measuring how strong the payoff difference influences the imitation probability \cite{Arefin_2021,hofbauer1998evolutionary,wu2011evolutionary,Farahbakhsh2021}. 
A natural question arises: Can the way of imitation, i.e., the imitation function, play any role in the long term behavior in eco-evolutionary dynamics?
Motivated by this, we propose a class of eco-evolutionary dynamics, in which strategies coevolve with the environment with general imitation function. 
We analyze how the dynamical outcomes are changed by the general non-linear imitation rules.

\section{Model}
Let's consider an infinite and well mixed population.
The eco-evolutionary dynamics consist of two parts: strategy evolution and environment dynamics.

\textbf{Strategy evolution:} The fraction of cooperators in the population is denoted as $x$.
An individual is selected at random and subsequently another individual
is randomly selected in the population.
Let's denote that the first individual uses strategy \emph{i} and the other adopts strategy \emph{j}.
Note that both strategies \emph{i} and \emph{j} are either cooperation(\emph{C}) or defection(\emph{D}).
Then
the first individual imitates the second individual's strategy with probability $g(\beta(\pi_j-\pi_i))$,
where $\pi_i$ and $\pi_j$ are the average payoffs for the first individual and the second individual, respectively.
Genarally, $g$ is an imitation function, which fulfills two properties \cite{wu2011evolutionary, mcavoy2015}:\\
$(i)$ $g$ falls into internal $[0,1]$ with $g(x)>0$ for any $x \in R$, since $g$ is a probability.\\
$(ii)$ $g$ is a strictly increasing and continuously differentiable function.\\
In addition, the imitation function $g(x)$ can be generated by any distribution function of a random variable.
Non-negative parameter $\beta$
denotes the imitation intensity, measuring how strong
the payoff difference between the first and the second individual influences the imitation probability.
When imitation intensity $\beta$ is high, i.e., $\beta\to +\infty$, it implies that focal individual is more likely to learn the model individual's strategy, even though the opponent's payoff is slightly higher than the focal individual.
Noteworthily, the selection intensity mirrors the inverse temperature in statistical physics.

\textbf{Environment dynamics:} 
We denote the environmental state as $\emph{n}$, which is confined in $[0,1]$.
Consider an environment-dependent payoff matrix $A(n)$, which interpolates between two payoff matrixes:
\begin{equation}\label{1}
	A(n)=(1-n)
	\begin{bmatrix}
			R_0 & S_0 \\
			T_0 & P_0 \\
		\end{bmatrix}
	+n
	\begin{bmatrix}
			R_1 & S_1 \\
			T_1 & P_1 \\
		\end{bmatrix},
\end{equation}
where the restrictions of $R_0>T_0,S_0>P_0,$ $R_1<T_1, S_1<P_1$ are satisfied.
And the payoffs of cooperators and defectors are:
\begin{equation*} 
	\pi_C(x,n)=((R_1-R_0-S_1+S_0)n+(R_0-S_0))x+(S_1-S_0)n+S_0,
\end{equation*} 
\begin{equation*}
	\pi_D(x,n)=((T_1-T_0-P_1+P_0)n+(T_0-P_0))x+(P_1-P_0)n+P_0.
\end{equation*} 
When $n = 1$, the environment is replete and ``tragedy of the commons'' occurs.
The game dominated by defection strategy is played ($R_1<T_1, S_1<P_1$).
On the contrary, when $n = 0$, it refers to depleted and the game dominated by cooperation strategy is played ($R_0>T_0, S_0>P_0$).
It implies that mutual cooperation is a Nash equilibrium under the depleted and defection is a Nash equilibrium under the replete environment.
We adopt the payoff matrix Eq.~\ref{1} and assume that $R_0>T_0, S_0>P_0, R_1<T_1, S_1<P_1$ so that the mutual cooperation and defection are Nash equilibria
under the depleted environment and the replete environment, respectively.

We assume that: $(i)$ The more cooperators there are, the better environment it is;
$(ii)$ The more defectors there are, the worse environment it is.
In one word, it is a type of feedback in which cooperators enhance the environment, replete environment breeds defection, then
defectors destroy the environment and depleted environment is conductive to evolution of cooperators.
Thus the imitation-environment coevolutionary model reads
\begin{equation}\label{4}
    \begin{cases}
        \dot{x}=x(1-x)(g(\beta(\pi_C-\pi_D))-g(\beta(\pi_D-\pi_C)))\\
        \dot{n}=\epsilon n(1-n)(-1+(1+\theta)x)
    \end{cases},
\end{equation}
where $\epsilon>0$ is the relative evolution speed, and $\theta>0$ is the relative speed of enhancement rates compared to degradation rates of cooperators and defectors, respectively. 
In particular, if $\frac{g(\beta(\pi_C-\pi_D))-g(\beta(\pi_D-\pi_C))}{\pi_C-\pi_D}$ is considered as time scale, the resulting dynamics is classical eco-evolutionary dynamics. 
A cooperator encounters a defector with probability $x(1-x)$.
Then the cooperator imitates the defector's strategy with a probability of imitation functions, including but not limited to the well-known Fermi function \cite{traulsen2010human}.
The logistic term $n(1-n)$ ensures that the environment is restrained to $[0,1]$. 
In addition, the term $(-1+(1+\theta) x)$ describes the environmental feedback mechanism with the assumption that the more cooperators there are, the richer the environment is. 

To compare with the classical model of coevolutionary dynamics used replicator dynamics framework which is given by \cite{weitz2016oscillating}, we show the mathematical model:
\begin{align}
	\dot{x}&=x(1-x)(\pi_C(x,n)-\pi_D(x,n)),\label{2} \\
	\dot{n}&=\epsilon n(1-n)(-1+(1+\theta)x),\label{3}
\end{align}
Here $\epsilon>0$ is the relative evolution speed and
$\theta$ is the relative speed of enhancement rates compared to degradation rates of cooperators and defectors, respectively.
To study the tragedy of the commons, mutual defection is perceived as Nash equilibrium under the replete environment ($n=1$).
As for the depleted environment ($n=0$), it has different Nash equilibria cases.
For example, mutual cooperation is a Nash equilibrium under the depleted and defection is a Nash equilibrium under the replete environment.

\section{Results}
We demonstrate that the replicator equation described by Eq.~\ref{2} is equilibrant to the strategic evolution equation described by the first equation in Eq.~\ref{4}, provided that the environment is not changing.
In fact, the right-hand side of the first equation in Eq.~\ref{4} is the product of the right-hand side of Eq.~\ref{2} and $\frac{g(\beta(\pi_C-\pi_D))-g(\beta(\pi_D-\pi_C))}{\pi_C-\pi_D}$, which is defended as $B(x,n)$.
Since $g$ is strictly increasing, $B(x,n)$ is positive.
Thus the solutions of the differential Eq.~\ref{2} and Eq.~\ref{4} can be mapped to each other by time rescaling \cite{hofbauer1998evolutionary}.
Noteworthily, this is true for all the imitation function $g$.
It implies that imitation functions have no qualitative effects on the replicator dynamics for static environment.
Then a natural question arises:
does the imitation function $g$ have any qualitative impact on the eco-evolutionary dynamics, in which the evolution of strategy is coupled with the environment?

\subsection{The local stability of the equilibria \textit{cannot} be altered by the imitation mode}
When the parameters satisfies conditions $R_0>T_0,S_0>P_0,R_1<T_1, S_1<P_1$ and $\theta>0$, there are five fixed points in the imitation-environment coevolution dynamics for Eq.~\ref{4}: \\
$(i)$ $(x^*=0,n^*=0)$, all individuals adopt defection in the depleted environment;\\
$(ii)$ $(x^*=1,n^*=0)$, all individuals adopt cooperation in the depleted environment;\\
$(iii)$ $(x^*=0,n^*=1)$, all individuals adopt defection in the replete environment;\\
$(iv)$ $(x^*=1,n^*=1)$, all individuals adopt cooperation in the replete environment.\\
$(v)$ $(x^*=\frac{1}{1+\theta},n^*=\frac{(T_0-R_0)+\theta(P_0-S_0)}{(R_1-T_1+T_0-R_0)+\theta(S_1-P_1+P_0-S_0)})$, cooperators and defectors coexist in an intermediate environment. 
Noteworthily, the fifth equilibrium is not dependent on the imitation function.\\
The first four fixed points are on the ``boundary'', and the fifth is an interior fixed point. 
A topological circle of the four ``boundary'' fixed points and connecting heteroclinic orbits is called heteroclinic cycle (See Fig.~\ref{heteroclinic cycle}).
\begin{figure}
	\begin{minipage}{\linewidth} 
	\centerline{\includegraphics[scale=0.25]{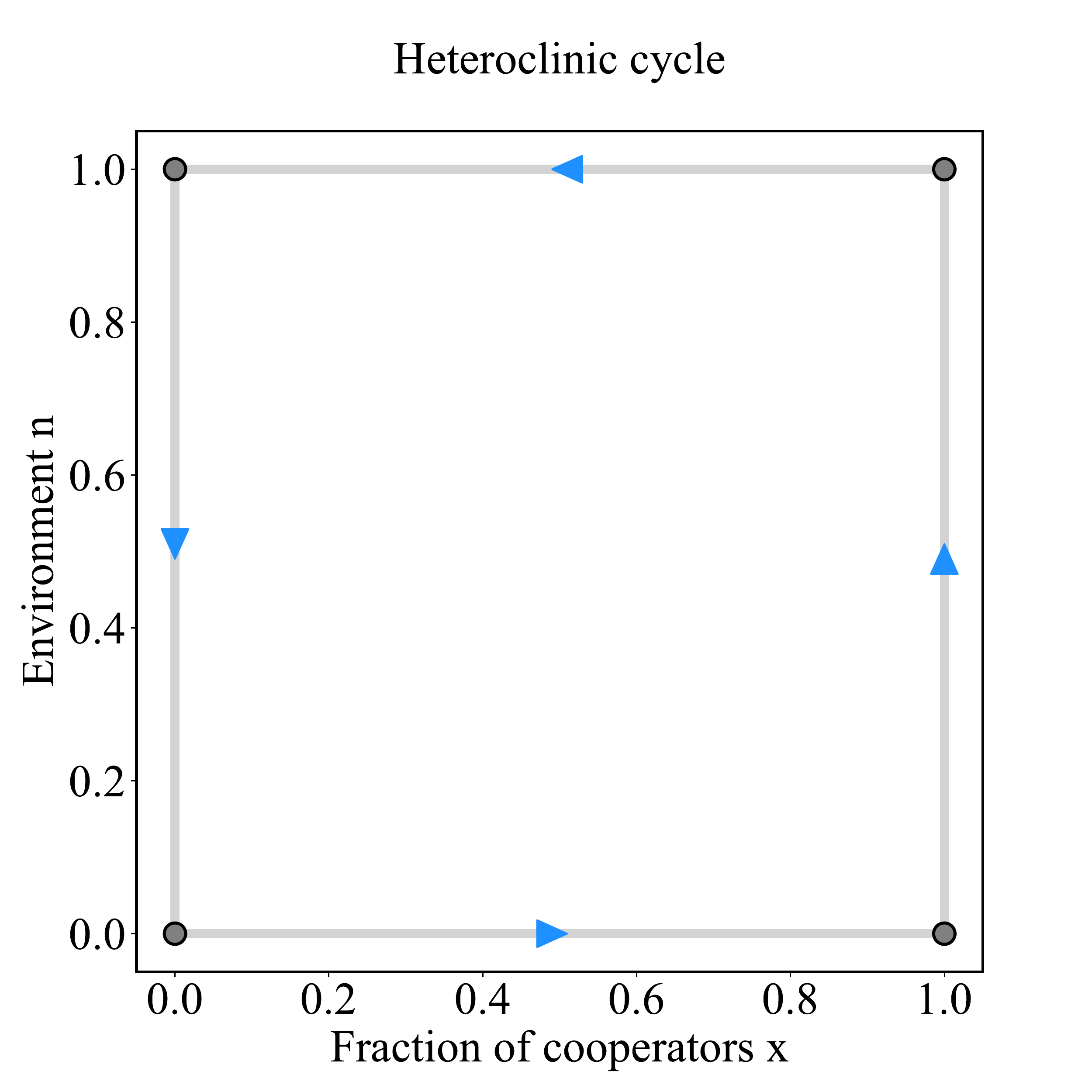}}
	\end{minipage}
	\caption{\textbf{Heteroclinic cycle in eco-evolutionary dynamics with general imitations.} Consider the coevolutionary dynamics given by Eq.~\ref{4}, the heteroclinic cycle incorporates four ``boundary'' fixed points $(0,0)$, $(1,0)$, $(0,1)$ and $(1,1)$ and four heteroclinic orbits.
	The heteroclinic cycle cycles counterclockwise.
	If the heteroclinic cycle is stable, it refers to an oscillating tragedy of commons.
	This heteroclinic cycle is present in classical eco-evolutionary model, and is also present for eco-evolutionary dynamics with the imitation function $g$.
	}\label{heteroclinic cycle} 
\end{figure}
The stability of these fixed points are given by Theorem.~\ref{Theorem.1}.

\begin{theorem}\label{Theorem.1}
 (Local stability) Consider the coevolutionary dynamics given by Eq.~\ref{4}, then,

\noindent $(i)$ Four `boundary'' fixed points are unstable.

\noindent $(ii)$ The interior fixed point is stable if
\begin{equation}\label{5}
\frac{P_1-S_1}{T_1-R_1}<\frac{S_0-P_0}{R_0-T_0}.
\end{equation}

\noindent $(iii)$ The interior fixed point is unstable if
\begin{equation}\label{6}
\frac{P_1-S_1}{T_1-R_1}>\frac{S_0-P_0}{R_0-T_0}.
\end{equation}
%
\end{theorem}

\begin{proof}
	The proof of Theorem.~\ref{Theorem.1} is shown in \ref{fixed points}.
\end{proof}

\textbf{Remark 1:} The local stability of internal fixed point is the same as that in the classical eco-evolutionary dynamics (Eq.~\ref{2}).
It is clear from the Theorem.~\ref{Theorem.1} that the stability condition of the internal equilibrium has nothing to do with the imitation function. 
In other words, the local stability of the internal fixed point \textit{cannot} be altered as long as individuals imitate based on the non-linear imitation function.

\subsection{The global eco-evolutionary dynamics can be sensitive to imitation functions}
Next, we investigate the global dynamics of Eq.~\ref{4}.
We first give the definitions of repelling and asymptotically stability of a heteroclinic cycle \cite{josef1994},
then we explore the stability of heteroclinic cycle, which incorporates four ``boundary'' fixed points $(0,0)$, $(1,0)$, $(0,1)$ and $(1,1)$ and four heteroclinic orbits (Fig.~\ref{heteroclinic cycle}).
The global dynamical behaviors in two-dimensional system depend not only on the local stability of the fixed points,
but also on the stability of the heteroclinic cycle.

\begin{definition}
	\label{def:1}
	Consider $\dot{x}=f(x)$ as an ODE defined on an open set $G\subseteq R^2$. The $\omega-$limit set of a point $x$ is the set of points that the flow from $x$ tends to in positive time, and the $\alpha-$limit set is in negative time.
\end{definition}

\begin{definition}
	\label{def:2}
	A heteroclinic cycle is an invariant topological circle $X$ consisting of the union of a set of equilibria {$\gamma_1$, $\dots$, $\gamma_k$} and orbits {$F_1$, $\dots$, $F_k$}, where $F_i$ is a heteroclinic orbit between $\gamma_i$ and $\gamma_{i+1}$; and $\gamma_{m+1} \equiv \gamma_1$.
	If $k = 1$ then the single equilibrium and connecting orbit form a homoclinic cycle.
\end{definition}

\begin{lem}
	\label{lem:01}
	\cite{10.1145/3372885.3373833,zbMATH02707188} The $\omega-$limit set, denoted as $\omega(x)$, is a closed set.
\end{lem}

\begin{lem}
 \label{lem:02}
 \cite{josef1994,M1997,Steindl1996HeteroclinicCI}
 \noindent $(i)$ The heteroclinic cycle $\Lambda$ is repelling when the following conditions are satisfied:
\begin{itemize}
  \item The heteroclinic cycle $\Lambda$ is a compact invariant subset of the boundary $\partial X$.
\end{itemize}
\begin{itemize}
  \item There is no $x\in int\,X$ with $\omega(x)\in\Lambda$.
\end{itemize}
\begin{itemize}
  \item There is a neighbourhood $U$ of $\Lambda$ in $X$ such that for each $x\in U \verb+\+ \partial X$,
  $\exists t\in R: x(t)\notin U$. ($\partial X$ is isolated near $\Lambda$.)
\end{itemize}

\noindent $(ii)$ The heteroclinic cycle $\Lambda$ is asymptotically stable whenever both of the following two conditions are satisfied:
\begin{itemize}
  \item The heteroclinic cycle $\Lambda$ is a compact invariant subset of the boundary $\partial X$.
\end{itemize}
\begin{itemize}
  \item Let $\Lambda_{(k)}\subset \Lambda$ for $k=1,\cdots, m$ ($m$ is the number of the fixed point in $\Lambda$) be
  such that for each $x\in \Lambda$, there is a $k$ with $\omega(x)\in \Lambda_{(k)}$ ($\omega(x)$ is the $\omega$-limit set of $x$).
\end{itemize}
\end{lem}

\begin{lem}
\label{lem:03}
(Poincar\'{e}-Bendixson Theorem)
Let $\dot{x}=f(x)$ be an ODE defined on an open set $G\subseteq R^2$.
Let $\omega(x)$ be a nonempty compact $\omega$-limit set.
Then, if $\omega(x)$ contains no rest point,
it must be a periodic orbit.
\end{lem}

\begin{theorem}\label{Theorem.2}
	(Existence of periodic orbit) If
	\begin{equation}\label{6}
		\frac{P_1-S_1}{T_1-R_1}>\frac{S_0-P_0}{R_0-T_0},
	\end{equation}
	and
	\begin{equation}\label{9}
		\begin{split}
			&[g(\beta(S_0-P_0))-g(\beta(P_0-S_0))]\times[g(\beta(T_1-R_1))-g(\beta(R_1-T_1))]>\\
			&[g(\beta(S_1-P_1))-g(\beta(P_1-S_1))]\times[g(\beta(T_0-R_0))-g(\beta(R_0-T_0))]
		\end{split}
	\end{equation}
	then there is a periodic orbit for the system given by Eq.~\ref{4}.
\end{theorem}

\begin{proof}
	Here, we provide a sketch of the proof. 
To prove the existence of the periodic orbit, we need to find a nonempty compact $\omega-$limit set, which contains no rest point (Lemma.~\ref{lem:03}). 
The internal equilibrium $x^*$ is unstable and the heteroclinic cycle is repelling since inequalities.~\ref{6} and \ref{9} hold (See \ref{fixed points} and \ref{boundary cycle} for the more detailed proof). 
There exists a $x \in int[0,1]^2$ such that the internal equilibrium $x^* \notin \omega(x)$ and $\omega(x) \cap bd[0,1]^2 = \emptyset$. 
In addition, inequality.~\ref{6} implies that $\omega(x)$ contains no rest point. 
Thus, Theorem.~\ref{Theorem.2} is proved (See \ref{Existence of periodic orbit} for detailed proof).
\end{proof}

\begin{definition}
	\label{def:3}
	The isolated closed orbit is defined as the limit cycle.
\end{definition}

\begin{cor}
	\label{cor:1}
	For the system given by Eq.~\ref{4}, there is a stable limit cycle if the internal fixed point is unstable and the heteroclinic cycle is repelling (Inequalities.~\ref{6},~\ref{9} hold).
\end{cor}
\begin{proof}
	The detailed proof of Corollary.~\ref{cor:1} is shown in \ref{Stable limit cycles}.
\end{proof}

\textbf{Remark 2:} (Estimation of the limit cycle position.)
The proof of Corollary.~\ref{cor:1} also provides a way to estimate the position of the stable limit cycle.
Let us denote $D=\{x\in[0,1]\times [0,1]| div f(x^*)\times dif f(x)>0\}$, in which $x^*$ is the unique equilibrium.
Then the limit cycle cannot be within $D$.
In fact, there is no limit cycle in $D$.
If it is not true, then there exits a limit cycle $\gamma$,
whose interior is denoted $\Gamma\in D$.
Based on the Green theorem, we have that
$\int_{\Gamma} div f(x) d(x_1,x_2)=\pm\int_0^T[f_2(x)\dot{x}_1(t)-f_1(x)\dot{x}_2(t)]dt$,
where $T$ is the period of limit cycle $\gamma$.
The left hand side is negative or positive, while the right hand side is zero since $f_i=\dot{x}_i$, $i=1,2$,
a contradiction.
Therefore the limit cycle (which is proved to be present)
has to be outside $D$.
This could facilitate us to find the limit cycle.

\begin{lem}
	\label{lem:04}
	\cite{10.1145/3372885.3373833,zbMATH02707188} Consider $\dot{x}=f(x)$ as an ODE defined on an open set $G\subseteq R^2$. The $\alpha-$limit set of a point $x$ is the set of points that the flow from $x$ tends to in negative time, which denoted as $\alpha(x)$, is also a closed set.
\end{lem}

\begin{cor}
	\label{cor:2}
	For the system given by Eq.~\ref{4}, there is an unstable limit cycle if
		\begin{equation}\label{5}
		\frac{P_1-S_1}{T_1-R_1}<\frac{S_0-P_0}{R_0-T_0},
	\end{equation}
	and 
	\begin{equation}\label{10}
		\begin{split}
			&[g(\beta(S_0-P_0))-g(\beta(P_0-S_0))]\times[g(\beta(T_1-R_1))-g(\beta(R_1-T_1))]<\\
			&[g(\beta(S_1-P_1))-g(\beta(P_1-S_1))]\times[g(\beta(T_0-R_0))-g(\beta(R_0-T_0))]
		\end{split}.
	\end{equation}
\end{cor}
\begin{proof}
	The detailed proof of Corollary.~\ref{cor:2} is shown in \ref{Unstable limit cycles}.
\end{proof}

To further show the difference between dynamical outcomes of the eco-evolutionary dynamics with general imitation and that of the classical eco-evolutionary dynamics, we give the phase diagrams when the limit cycle is unstable (Fig.~\ref{Fig.3}) and stable (Fig.~\ref{Fig.4}).
In addition, numerical simulation (Fig.~\ref{Fig.5}) are performed to show that there is a limit cycle for the system given by Eq.~\ref{4} under different imitation functions.
\begin{figure}
	\begin{minipage}{0.5\linewidth} 
	\centerline{\includegraphics[]{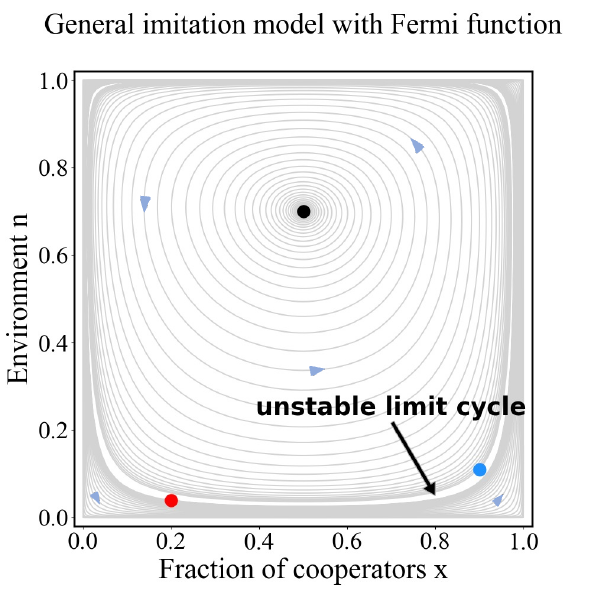}}
	\end{minipage}
	\begin{minipage}{0.5\linewidth} 
	\centerline{\includegraphics[]{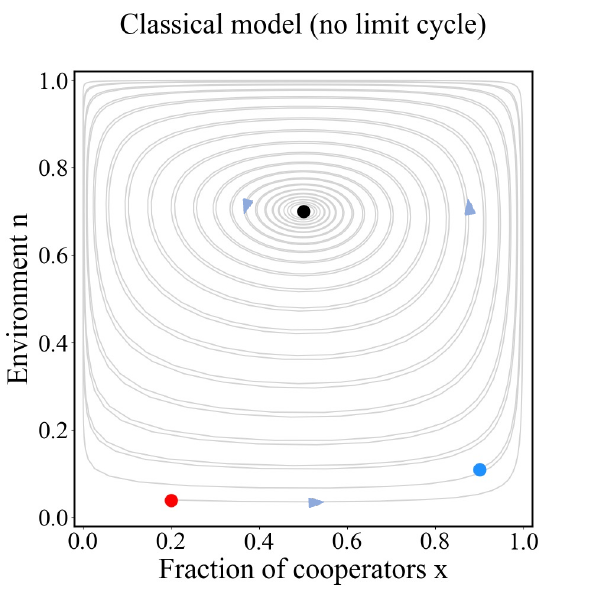}}
	\end{minipage}
	\caption{\textbf{Differences between eco-evolutionary dynamics with general imitation function and the classical eco-evolutionary dynamics when the internal fixed point is stable.}
		Phase diagrams of x-n in eco-evolutionary dynamics with general imitation function (Left). 
		Diagrams of x-n in the classical eco-evolutionary dynamics (Right). 
		The red solid circle and blue solid circle represent the initial points (0.2, 0.04), (0.9, 0.11), respectively. 
		Imitation function leads to an unstable limit cycle, which repels its inner trajectories to the vicinity of the internal fixed point, and repels its outer trajectories to the vicinity of the heteroclinic cycle. 
		In other words, different initial points lead the system to different steady states. 
		However, the internal fixed point is always stable regardless of the initial points in the classical model. 
		Parameters for both graphs are $R_0=6, S_0=6, T_0=4, P_0=1, R_1=2, S_1=4, T_1=3, P_1=6, \epsilon=1, \theta=1, \beta=1$.
	}\label{Fig.3}
\end{figure}

\begin{figure}
	\begin{minipage}{0.5\linewidth} 
	\centerline{\includegraphics[]{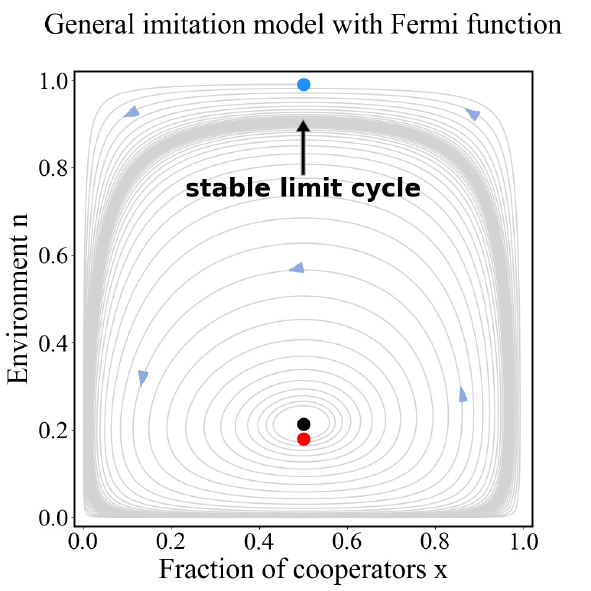}}
	\end{minipage}
	\begin{minipage}{0.5\linewidth} 
	\centerline{\includegraphics[]{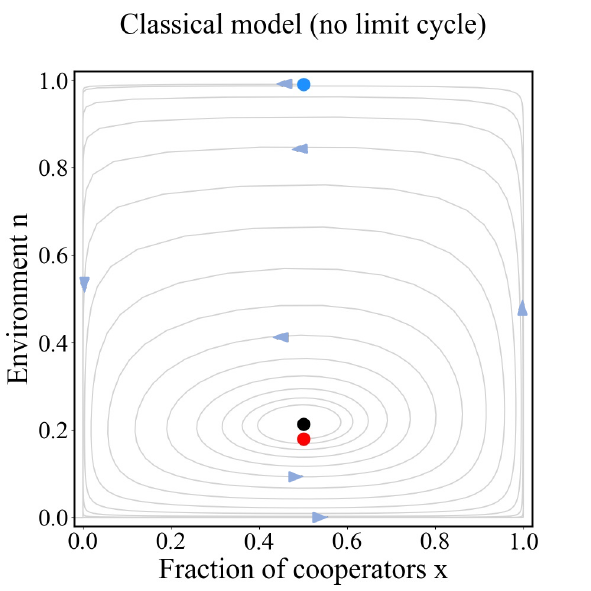}}
	\end{minipage}
	\caption{\textbf{Differences between eco-evolutionary dynamics with general imitation function and the classical eco-evolutionary dynamics when the internal fixed point is unstable.}
		Phase diagrams of $x$-$n$ in eco-evolutionary dynamics with general imitation function (Left).
		Phase plane of $x$-$n$ in the classical eco-evolutionary dynamics (Light).
		The red solid circle and blue solid circle represent the initial points $(0.5,0.2)$, $(0.5,0.99)$, respectively. 
		Imitation function leads to a stable limit cycle, which attracts trajectories to its surrounding. 
		However, the internal fixed point is always unstable regardless of the initial points in the classical model. 
		Parameters for both graphs are $R_0=2, S_0=9, T_0=1, P_0=7, R_1=6, S_1=1, T_1=9, P_1=9, \epsilon=1, \theta=1, \beta=1$.
	}\label{Fig.4}
\end{figure}

\begin{figure}
	\centering
	\begin{tiny}
	\begin{equation*} 
	~~~~g(x)=\frac{1}{1+e^{- x}}
	~~~~~~~~~~~~~~~~~g(x)=\frac{1}{\sqrt{2\pi}}\int_{-\infty}^{x}e^{-\frac{t^2}{2}}dt~~~~~~~~~~~~~~~~~g(x)=
	\begin{cases}
	 1-e^{- x}, x>0\\
	 0, x \leq 0
	 \end{cases}
	\end{equation*}
	\end{tiny}\\
	\setlength{\parskip}{-1em}
	\begin{minipage}{0.3\linewidth} 
	\centerline{\includegraphics[width=1\linewidth]{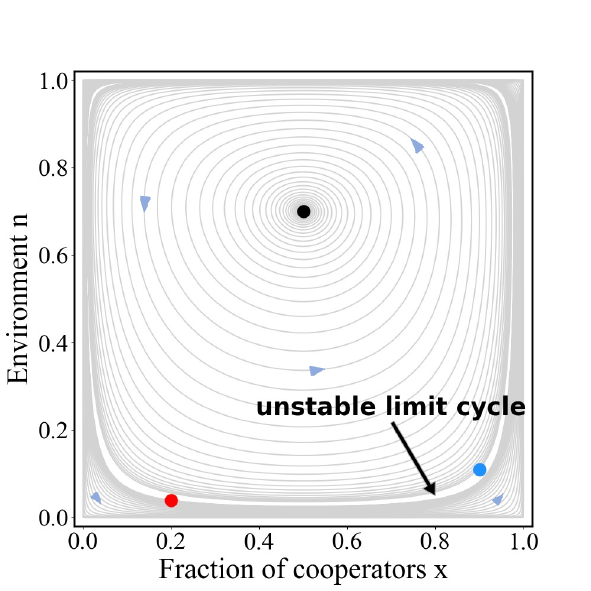}}
	\end{minipage}
	\begin{minipage}{0.3\linewidth} 
	\centerline{\includegraphics[width=1\linewidth]{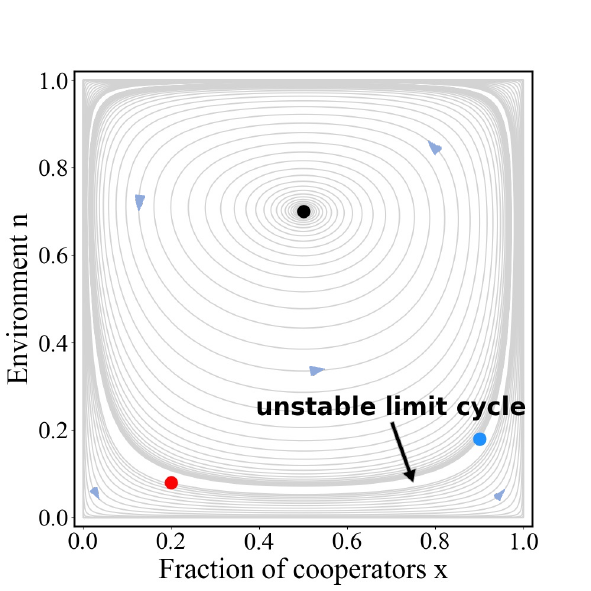}}
	\end{minipage}
	\begin{minipage}{0.3\linewidth} 
	\centerline{\includegraphics[width=1\linewidth]{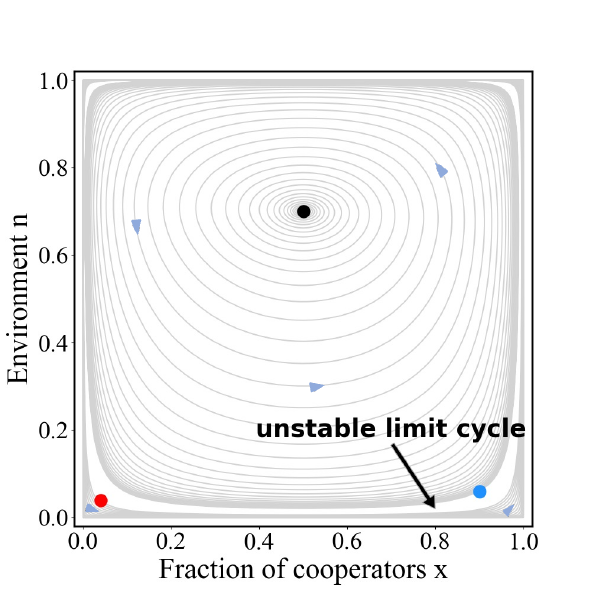}}
	\end{minipage}\\
	\begin{minipage}{0.3\linewidth} 
	\centerline{\includegraphics[width=1\linewidth]{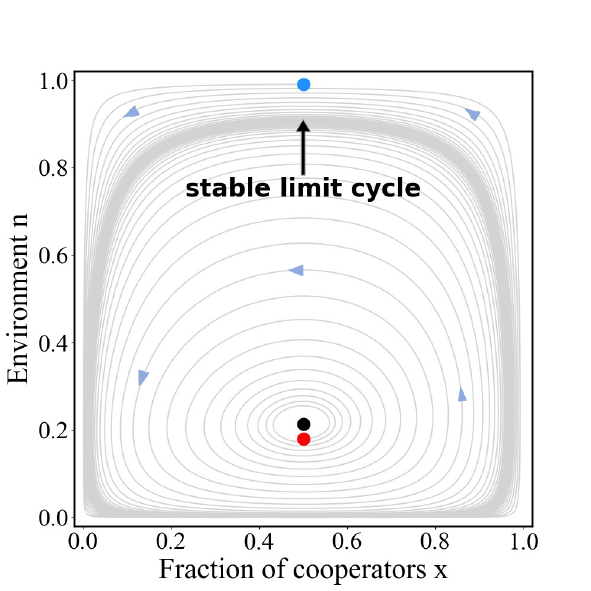}}
	\end{minipage}
	\begin{minipage}{0.3\linewidth} 
	\centerline{\includegraphics[width=1\linewidth]{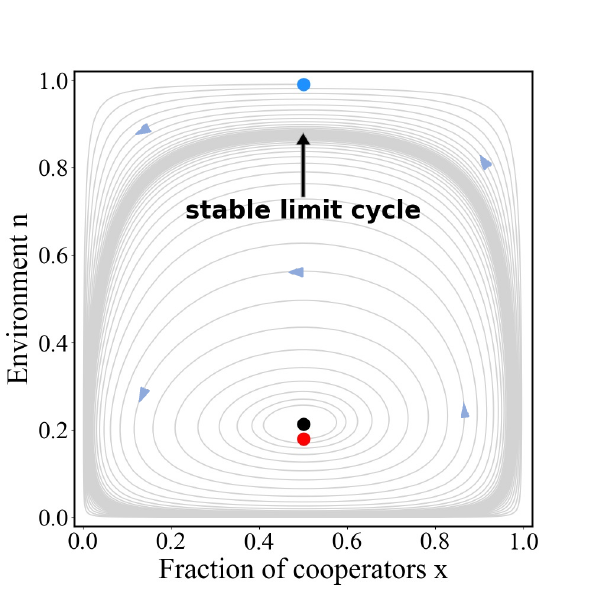}}
	\end{minipage}
	\begin{minipage}{0.3\linewidth} 
	\centerline{\includegraphics[width=1\linewidth]{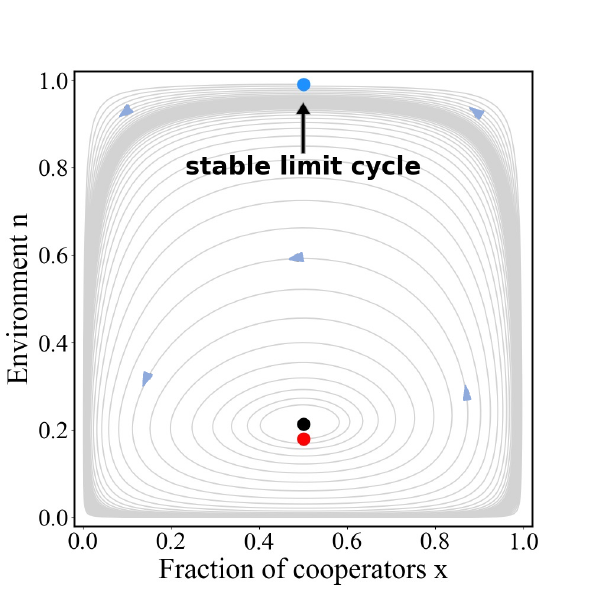}}
	\end{minipage}
	\caption{\textbf{Existence of limit cycles for eco-evolutionary dynamics.}
		Each column has the same imitation function, which is shown as the top of figure.
		There are an unstable limit cycle in the first row and a stable limit cycle in the second row under three different imitation functions.
		The red and blue circles denote the initial points, and the black circle is the internal equilibria.
		The equilibria is stable in the first row and it is unstable in the second row.
		Parameters are the same as in Fig.~\ref{Fig.3} and Fig.~\ref{Fig.4}, respectively.
	}\label{Fig.5}
\end{figure}

\subsection{Full classification of dynamical outcomes}
So far, we have shown that the global eco-evolutionary dynamics can be sensitive to imitation functions, although the local stability of the internal equilibrium is robust to imitation functions.
Here, to further understand how the imitation function alters the global dynamics, we distinguish the dynamical outcomes of the system into four cases based on the stability of both the internal fixed point and heteroclinic cycle (See Fig.~\ref{summary}):
\begin{figure}
	\centering
	\begin{tiny}
	\begin{equation*} 
	\begin{split}
		&[g(\beta(S_0-P_0))-g(\beta(P_0-S_0))]\times[g(\beta(T_1-R_1))-g(\beta(R_1-T_1))]>\\
		&[g(\beta(S_1-P_1))-g(\beta(P_1-S_1))]\times[g(\beta(T_0-R_0))-g(\beta(R_0-T_0))]
	\end{split}
	\end{equation*}
	\end{tiny}\\
	\centering
	\begin{minipage}{0.45\linewidth} 
	\centerline{\includegraphics[width=1\linewidth]{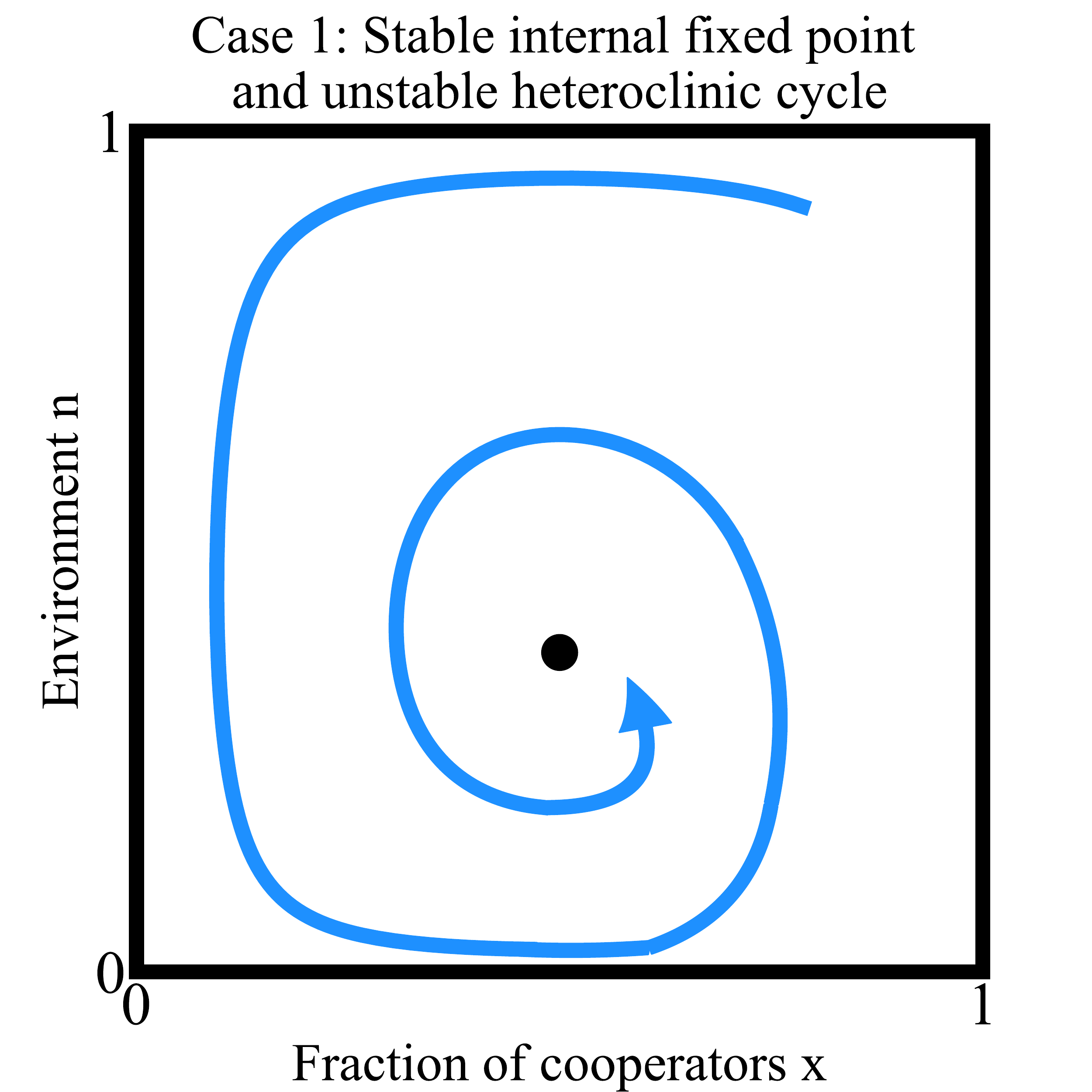}}
	\end{minipage}
	\begin{minipage}{0.45\linewidth} 
	\centerline{\includegraphics[width=1\linewidth]{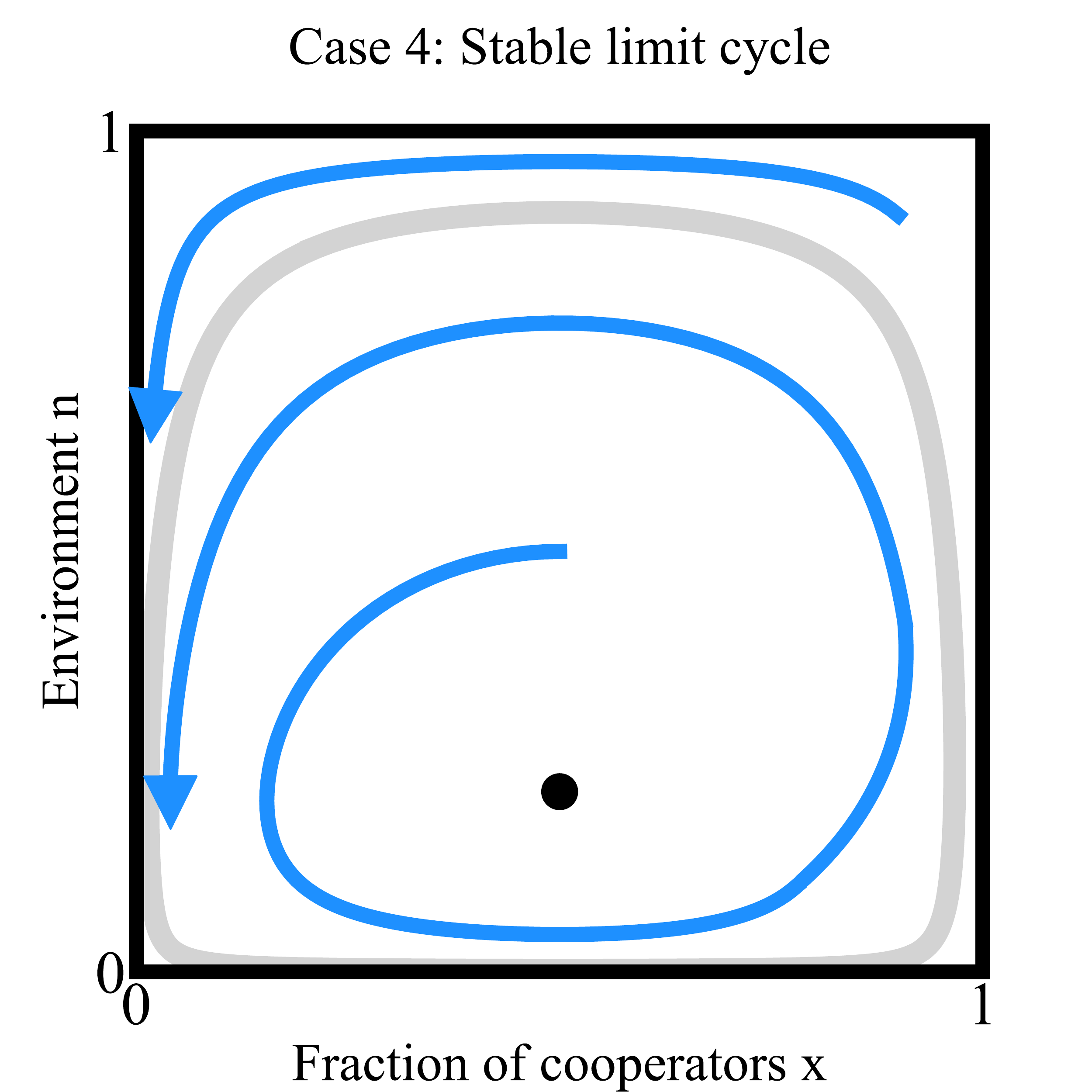}}
	\end{minipage}\\
	\vspace{-16em}
\begin{tikzpicture}
\filldraw  [black](0,0)rectangle(12,0.05);
\filldraw  [black](6,0)rectangle(6.05,6);
\filldraw  [black](6,0)rectangle(6.05,-6);
\end{tikzpicture}
	\vspace{-18em}
	\begin{tiny}
	\begin{equation*} 
	\frac{P_1-S_1}{T_1-R_1}>\frac{S_0-P_0}{R_0-T_0}
~~~~~~~~~~~~~~~~~~~~~~~~~~~~~~~~~~~~~~~~~~~~~~~~~~~~~~~~~~~~~~~~~~~~~~~~~~~~~~~~~~~~~~~~~~~~~~~~
		\frac{P_1-S_1}{T_1-R_1}<\frac{S_0-P_0}{R_0-T_0}
	\end{equation*}
	\end{tiny}\\
	\vspace{-1em}
	\begin{minipage}{0.45\linewidth} 
	\centerline{\includegraphics[width=1\linewidth]{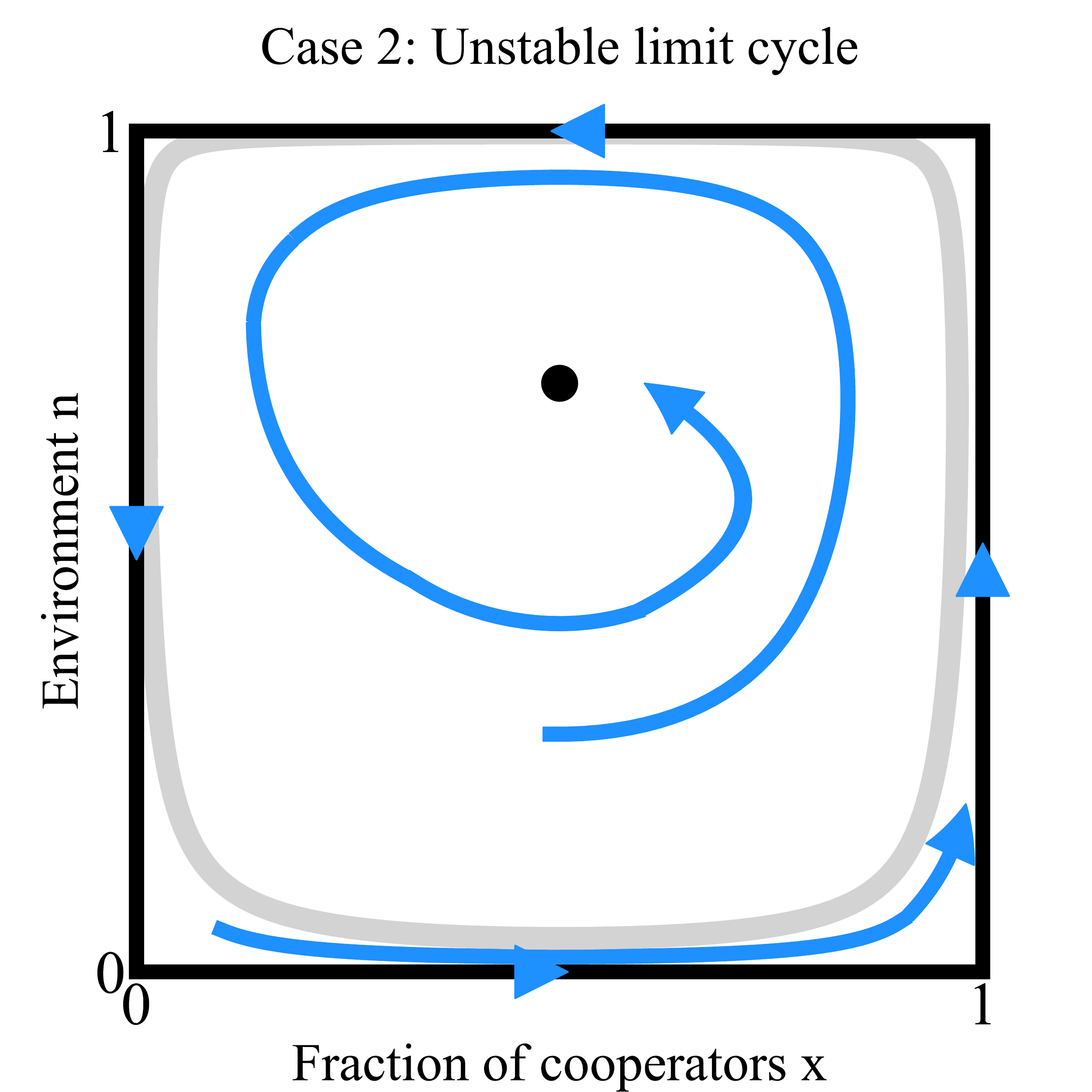}}
	\end{minipage}
	\begin{minipage}{0.45\linewidth} 
	\centerline{\includegraphics[width=1\linewidth]{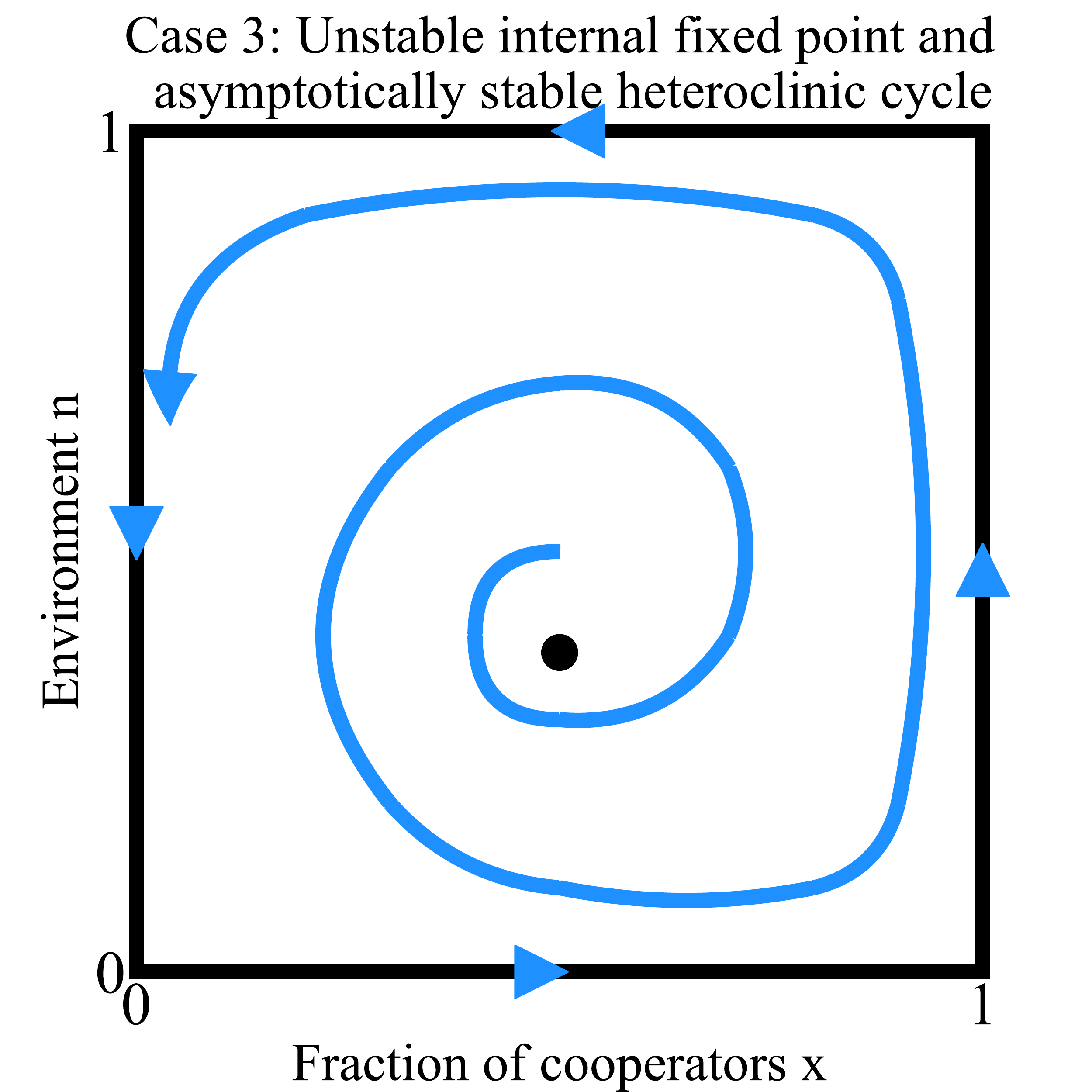}}
	\end{minipage}
	\begin{tiny}
	\begin{equation*} 
	\begin{split}
		&[g(\beta(S_0-P_0))-g(\beta(P_0-S_0))]\times[g(\beta(T_1-R_1))-g(\beta(R_1-T_1))]<\\
		&[g(\beta(S_1-P_1))-g(\beta(P_1-S_1))]\times[g(\beta(T_0-R_0))-g(\beta(R_0-T_0))]
	\end{split}
	\end{equation*}
	\end{tiny}
	\vspace{-3em}
	\caption{\textbf{Dynamics of general imitation function.}
	The black circle denotes the internal fixed point and gray line in case 2 and 4 refers to a limit cycle.
	There is an unstable limit cycle when the internal fixed point is stable and the heteroclinic cycle is asymptotically stable (case 2).
	And there is a stable limit cycle when the internal fixed point is unstable and the heteroclinic cycle is repelling (case 4).
	These two phenomenons (case 2 and 4) are never present in classical eco-evolutionary dynamics.
	The phenomena that the internal fixed point is stable and the heteroclinic cycle is repelling (case 1), and that the internal fixed point is unstable and the heteroclinic cycle is asymptotically stable (case 3), also occur in classical eco-evolutionary model.
	}\label{summary}
\end{figure}

\textbf{Case 1:} The internal fixed point is stable and the heteroclinic cycle is repelling (Inequalities.~\ref{9},~\ref{5} hold). 
In this case, the population becomes a stable mixture of cooperation and defection types.
The tragedy of the commons is averted in this case.
This phenomenon also occurs in classical eco-evolutionary dynamics when $\frac{P_1-S_1}{T_1-R_1}<\frac{S_0-P_0}{R_0-T_0}$ \cite{weitz2016oscillating}.

\textbf{Case 2:} The internal fixed point is stable and the heteroclinic cycle is asymptotically stable (Inequalities.~\ref{5},~\ref{10} hold).
Numerical simulations show that there is an unstable limit cycle in the state space (see Fig.~\ref{Fig.1}), which is consistent with the Corollary.~\ref{cor:2}.
The intuition is below: a fraction of trajectories is attracted by the internal fixed point and other fraction of trajectories flow the heteroclinic cycle.
Due to no bifurcation in two-dimensional systems, there must be unstable limit cycle in the system, which repels its inner trajectories to the vicinity of the internal fixed point, and repels its outer trajectories to the vicinity of the heteroclinic cycle.
The tragedy of the commons is also averted in this case.
However, this phenomenon can't occur in the classical eco-evolutionary dynamics.

\textbf{Case 3:} The internal fixed point is unstable and the heteroclinic cycle is asymptotically stable (Inequalities.~\ref{6},~\ref{10} hold).
In this case, the system finally moves close to the heteroclinic cycle and cycles counterclockwise, which is termed an oscillating tragedy of the commons.
The defectors increases when the environment is replete, and the increase of defectors leads to the deterioration of the environment, and vice versa.
This phenomenon also occurs in classical eco-evolutionary model when $\frac{P_1-S_1}{T_1-R_1}>\frac{S_0-P_0}{R_0-T_0}$ \cite{weitz2016oscillating}.

\textbf{Case 4:} The internal fixed point is unstable and the heteroclinic cycle is repelling (Inequalities.~\ref{6},~\ref{9} hold).
In this case, numerical simulations show that there is a stable limit cycle (see Fig.~\ref{Fig.2}).
The fixed point repels its neighbouring trajectories to its outside and the heteroclinic cycle repels its vicinal trajectories to its inside.
Therefore, there is a stable limit cycle in the system which attracts trajectories to its surrounding.
This is also the phenomenon of the oscillating tragedy of commons.

\begin{figure}
	\begin{minipage}{0.5\linewidth} 
	\includegraphics[scale=0.25]{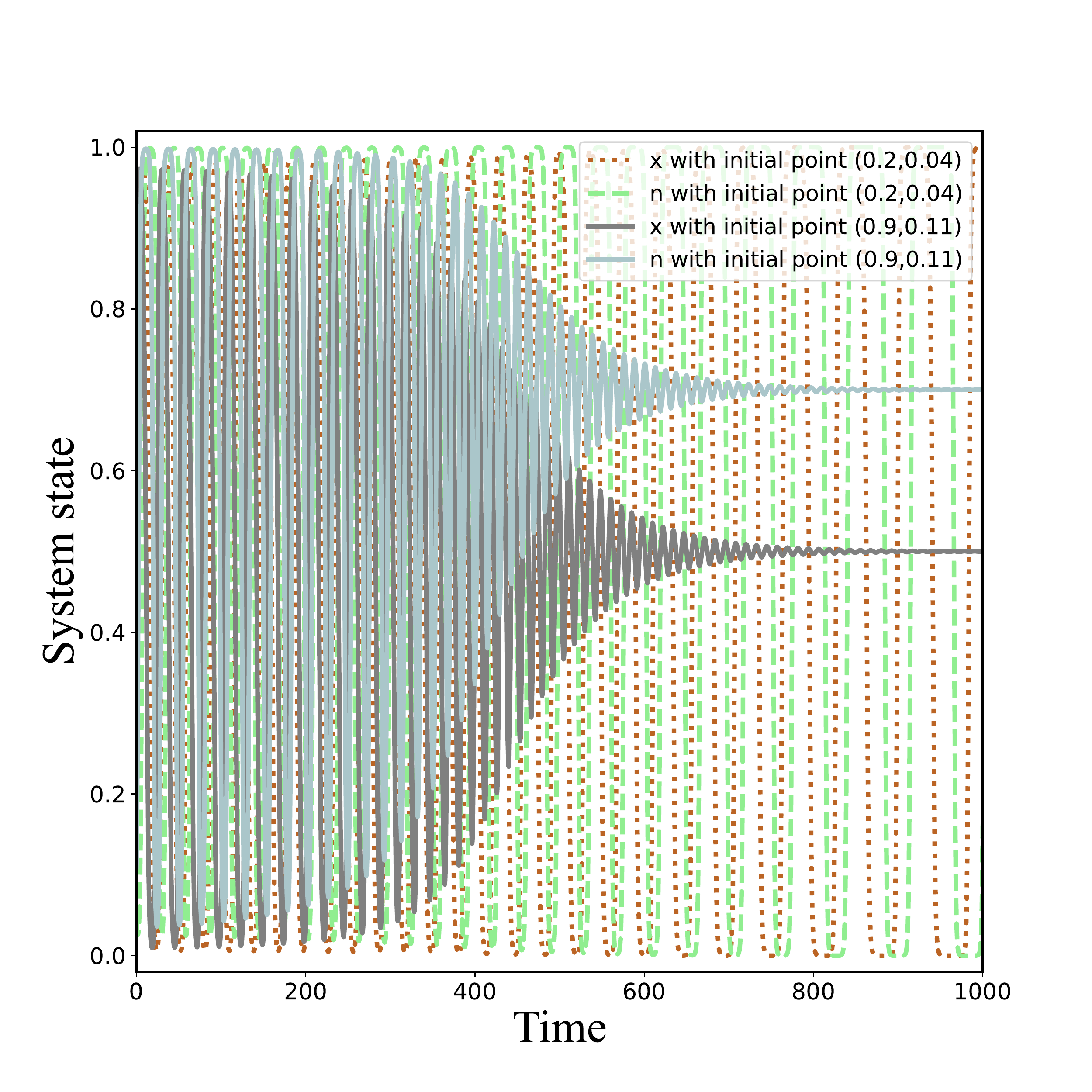}
	\end{minipage}
	\begin{minipage}{0.5\linewidth} 
	\includegraphics[scale=1]{fig4-1-1.pdf}
	\end{minipage}
\caption{\textbf{Existence of unstable limit cycles for eco-evolutionary dynamics.} 
Time series of the fraction of cooperators $x$ and the environment $n$ of eco-evolutionary dynamics with general imitation function (Left). 
Two solid lines and two dotted lines denote that the initial points is $(0.9,0.11)$ and $(0.2,0.04)$, respectively. 
Phase diagrams of the eco-evolutionary dynamics of $x$-$n$ (Right). 
The parameters are the same as Fig.~\ref{Fig.3}, and the imitation function is $g(\pm\beta(\pi_C-\pi_D))=\frac{1}{1+e^{\mp\beta(\pi_C-\pi_D)}}$.
}\label{Fig.1}
\end{figure}

\begin{figure}
	\begin{minipage}{0.5\linewidth} 
	\includegraphics[scale=0.25]{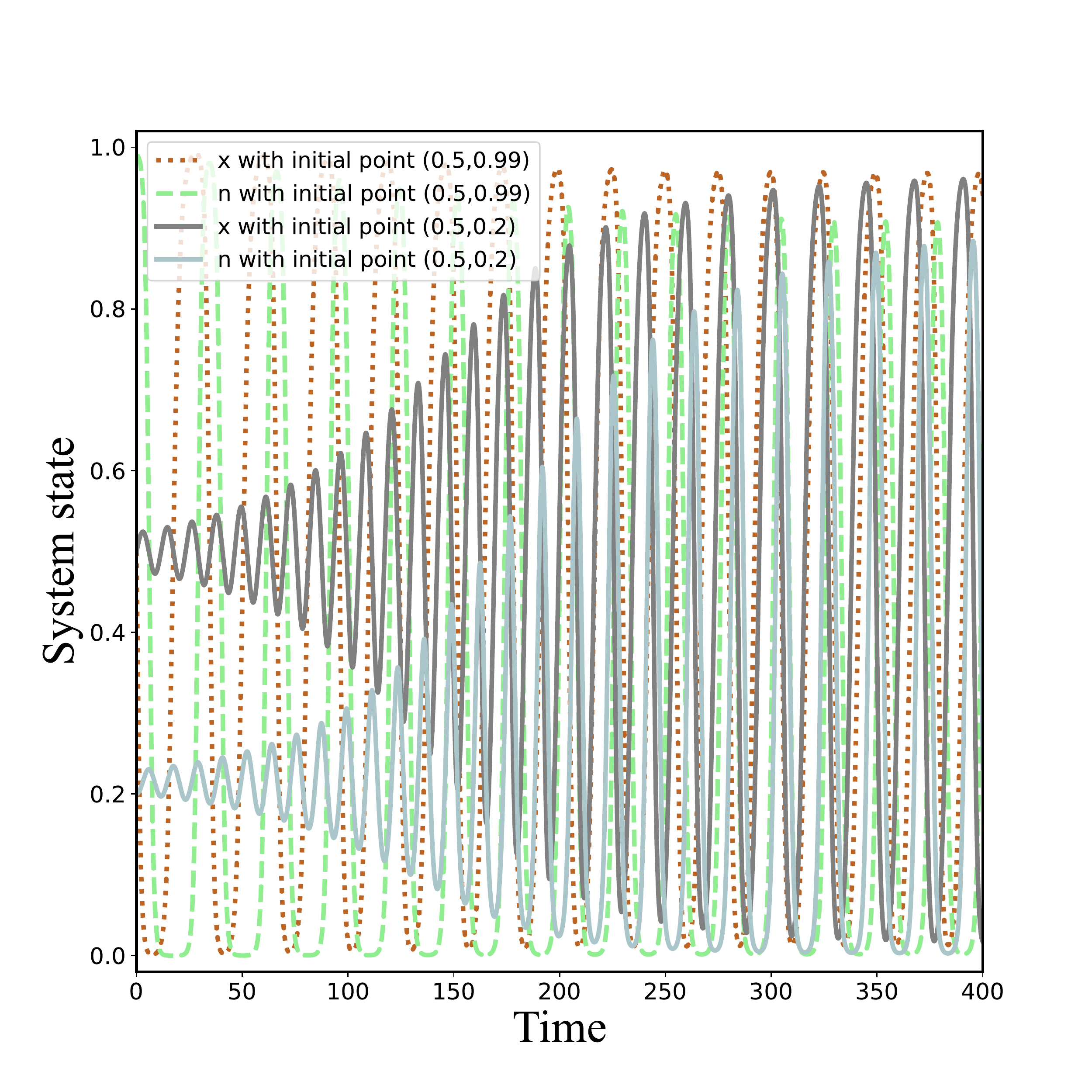}
	\end{minipage}
	\begin{minipage}{0.5\linewidth} 
	\includegraphics[scale=1]{fig4-1-2.pdf}
	\end{minipage}
\caption{\textbf{Existence of stable limit cycles for eco-evolutionary dynamics.}
Time series of the fraction of cooperators $x$ and the environment $n$ of eco-evolutionary dynamics with general imitation function (Left). 
Two solid lines and two dotted lines denote that the initial points is $(0.5,0.2)$ and $(0.5,0.99)$, respectively. 
Phase diagrams of the eco-evolutionary dynamics of $x$-$n$ (Right).
The parameters are the same as Fig.~\ref{Fig.4}, and the imitation function is $g(\pm\beta(\pi_C-\pi_D))=\frac{1}{1+e^{\mp\beta(\pi_C-\pi_D)}}$.
}\label{Fig.2}
\end{figure}

\subsection{Average cooperation level is approximately consistent with the first component of the internal fixed point}
The internal fixed point  $(x^*,n^*)$ is stable in case 1, which represents cooperators and defectors coexist in an intermediate environment. 
Regardless of the initial state of the system, $x^*$ and $n^*$ are the average cooperation level and the average environmental state, respectively, when the time is sufficiently long.
Stable limit cycles in case 4 describe the oscillatory coexistence of cooperators and defectors in the oscillating tragedy of commons. 
Similarly, we focus on the average cooperation level $\bar{x}$ and the average environment state $\bar{n}$ of the system when the time is sufficiently long. 
And we find that $\bar{x}\approx x^*$ and $\bar{n}\approx n^*$, in which $x^*=\frac{1}{1+\theta}$ and $n^*=\frac{(T_0-R_0)+\theta(P_0-S_0)}{(R_1-T_1+T_0-R_0)+\theta(S_1-P_1+P_0-S_0)}$ (See \ref{average} for detail).

\section{Discussion and Conclusions}
Eco-evolutionary dynamics is present, provided that the interplay between individuals' behaviors and the surrounding environment is taken into account.
Previous works on eco-evolutionary dynamic assumes that individuals imitate each other based on the linear imitation function, i.e., the replicator equation. 
The last decade has seen a generalization of the replicator equation by taking into account the dynamics of the environment. 
For example, Weitz et al. has proposed the replicator dynamics of evolutionary game with feedback mechanism \cite{weitz2016oscillating}. 
They assume that cooperators make the environment better and defectors make it worse.
However, alternative imitation rules other than replicator equations are neglected in eco-evolutionary dynamics.
In fact, imitation rule is widely adopted, which implies that an individual imitates other's strategy with a probability determined by the payoff difference. 
It is far from known how individuals imitate each other. 
And it has been unclear whether the replicator equation is the way individuals imitate each other up till now.
In other words, a replicator-like equation is way too specific. 

Motivated by this, we have extrapolated from this assumption, and have assumed that individuals imitate based on a general imitation function, which fulfills two natural assumptions. 
In this case, the replicator equation refers to the linear imitation function. 
It is noteworthy that the evolution of the strategy cannot be affected by the imitation rules in the sense of the global dynamics, provided that the environment remains unchanged.
Furthermore, we have compared the dynamical results with those of classical eco-evolutionary dynamics proposed by Weitz et al. \cite{weitz2016oscillating}: The fixed point and its local stability of the two eco-evolutionary dynamics are the same.
This implies that the local dynamics near the internal equilibrium cannot be altered by extrapolating from replicator to general imitation function.
The major difference lies in the stability of the heteroclinic cycle.
There is an unstable limit cycle when the internal fixed point is stable and the heteroclinic cycle is asymptotically stable,
and there is a stable limit cycle when the internal fixed point is unstable and the heteroclinic cycle is repelling in the eco-evolutionary dynamics driven by general imitation function.
These two cases \textit{cannot} occur in the classical eco-evolutionary dynamics.

The oscillatory coexistence is present ranging from biology to social systems and the limit cycles are found in lots of non-linear dynamics. 
Some classic approaches, such as Lasalle’s theorems \cite{la1976stability, ahmadi2007control, filippov2013differential, lasalle1976stability}, the Lyapunov stability theorem \cite{lyapunov1907general, reza2017robust, chellaboina2008nonlinear}, Poincaré maps \cite{hiskens2001stability, gonccalves2005regions} and Hamilton’s principle \cite{patil2001limit}, are widely used to prove the existence and the stability of the limit cycles. 
These approaches are rigorous but complex, because they are used to prove the stability of limit cycles from the global state of the dynamic system. 
Our approach, however, just needs to take into account the local stability of both the internal equilibrium and the heteroclinic cycles.
In Weitz et. al., one of their main results is that there cannot be any limit cycles in their model. 
It is assumed that individuals update their behavior via replicator equation, which results from the linear imitation function of payoff differences. 
However, the model is not sufficient to explain eco-evolutionary systems in which the internal period orbits are present. 
This suggests that additional mechanisms are required to explain the fluctuating behavior in biological and social systems \cite{Sanchez2013}. 
We prove that non-linear imitation is sufficient to produce the oscillatory dynamics.
Previous work has revealed that mutation \cite{gong2020limit, mobilia2010oscillatory}, multi-strategy \cite{mobilia2010oscillatory, hauert2002volunteering, peltomaki2008three, szolnoki2014cyclic} and empty space \cite{hauert2008ecological, peltomaki2008three, szolnoki2014cyclic} can be a way to such limit cycle. 
Our contribution is to highlight that the non-linearity in the imitation function alone is sufficient to create such stable oscillatory behavior, without mutation, multi-strategy nor empty space. 

To sum up, we have proposed a model of eco-evolutionary dynamics with the environment feedback mechanism based on general imitation functions.
We have answered what if individuals update their strategies via different imitation functions:
i) It does not change the position of the internal equilibrium and its local stability.
ii) It can lead to limit circle (both stable and unstable) for general imitation function.
The results are surprising that the global eco-evolutionary dynamics can be greatly altered by the way how individuals imitate.
Therefore, our work highlights the importance of the way of imitation in eco-evolutionary dynamics, which is typically absent in classical evolutionary dynamics with static environment.

\appendix
\section{\label{fixed points}Stability of fixed points}
There are five fixed points in the imitation-environment coevolution dynamics for Eq.~\ref{4}:
four of these are in ``boundary'', that is, \\
$(i)$ $(x^*=0,n^*=0)$, \\
$(ii)$ $(x^*=1,n^*=0)$, \\
$(iii)$ $(x^*=0,n^*=1)$,\\
$(iv)$ $(x^*=1,n^*=1)$,\\
and an interior fixed point \\
$(v)$  $(x^*=\frac{1}{1+\theta},n^*=\frac{(T_0-R_0)+\theta(P_0-S_0)}{(R_1-T_1+T_0-R_0)+\theta(S_1-P_1+P_0-S_0)})$.

We investigate the fixed points and analyze their stabilities with the aid of Jacobian matrices: 
i) if determinant of the matrix $det(J(x^*,n^*))$ is positive and trace of it $tr(J(x^*,n^*))$ is negative, the fixed point $(x^*,n^*)$ is stable;
ii)  if determinant $det(J(x^*,n^*))$ is positive and trace $tr(J(x^*,n^*))$ is positive, the fixed point $(x^*,n^*)$ is unstable;
iii) if determinant $det(J(x^*,n^*))$ is negative and trace $tr(J(x^*,n^*))$ is either zero or indeterminacy, the fixed point $(x^*,n^*)$ is a saddle.
For simplicity, denote the Jacobian matrix of the system described by Eq.~\ref{4} as $J$, which yields:
\begin{equation}\label{7}
	J=
	\left(
	\begin{array}{cc}
		\frac{\partial \dot{x}}{\partial x}&\frac{\partial \dot{x}}{\partial n}\\
		\frac{\partial \dot{n}}{\partial x}&\frac{\partial \dot{n}}{\partial n}\\
	\end{array}
	\right),
\end{equation}
where
\begin{equation}\nonumber
	\begin{split}
		\frac{\partial \dot{x}}{\partial x}_{(x^{*},n^{*})}
		&=(1-2x^*)(g(\beta(\pi_C-\pi_D))-g(\beta(\pi_D-\pi_C)))\\
		&+2\beta x^*(1-x^*)g'(\beta(\pi_C-\pi_D))\frac{\partial (\pi_C-\pi_D)}{\partial x}|_{(x^{*},n^{*})},\\
		\frac{\partial \dot{x}}{\partial n}_{(x^{*},n^{*})}
		&=2\beta x^*(1-x^*)g'(\beta(\pi_C-\pi_D))\frac{\partial (\pi_C-\pi_D)}{\partial n}|_{(x^{*},n^{*})},\\
		\frac{\partial \dot{n}}{\partial x}_{(x^{*},n^{*})}
		&=\epsilon n^{*}(1-n^{*})(1+\theta),\\
		\frac{\partial \dot{n}}{\partial n}_{(x^{*},n^{*})}
		&=\epsilon (1-2n^{*})(-1+(1+\theta)x^{*}).\\
	\end{split}
\end{equation}

So we have Jacobian matrices at four ``boundary'' fixed points:
\begin{equation}\nonumber
		J_{(0,0)}=\begin{bmatrix}
			g(\beta(S_0-P_0))-g(\beta(P_0-S_0)) & 0 \\
			0 & -\epsilon
		\end{bmatrix}; 
\end{equation}
\begin{equation}\nonumber
J_{(1,0)} = \begin{bmatrix}
			-(g(\beta(R_0-T_0))-g(\beta(T_0-R_0))) & 0 \\
			0 & \theta\epsilon
		\end{bmatrix};
\end{equation}

\begin{equation}\nonumber
	J_{(0,1)}=\begin{bmatrix}
		g(\beta(S_1-P_1))-g(\beta(P_1-S_1)) & 0 \\
		0 & \epsilon
	\end{bmatrix};
\end{equation}
\begin{equation}\nonumber
 J_{(1,1)} = \begin{bmatrix}
			-(g(\beta(R_1-T_1))-g(\beta(T_1-R_1))) & 0 \\
			0 & -\theta\epsilon
	        \end{bmatrix}.
\end{equation}

Recall that the function $g(\beta(\pi_j-\pi_i))$ is strictly increasing and $R_0>T_0,S_0>P_0,$ $R_1<T_1, S_1<P_1$, $\epsilon>0$, $\theta>0$.
All of the Jacobian matrixes all have one positive eigenvalue and one negative eigenvalue.
Therefore all the four ``boundary'' fixed points are unstable.

As for the Jacobian matrix at the internal fixed point $(x^*,n^*)$, denoted as $P$, we have

\begin{equation}\label{8}
	J|_{P}={
		\left[ \begin{array}{cc}
			2\beta x^{*}(1-x^{*})g'(0)\frac{\partial (\pi_C-\pi_D)}{\partial x}|_{P} & 2\beta x^{*}(1-x^{*})g'(0)\frac{\partial (\pi_C-\pi_D)}{\partial n}|_{P} \\
			\epsilon n^{*}(1-n^{*})(1+\theta) & 0\\
		\end{array}
		\right ]},
\end{equation}
where
\begin{equation}\nonumber
	\begin{split}
		\frac{\partial(\pi_C-\pi_D)}{\partial x}|_{P}&=\frac{(R_1-T_1)(P_0-S_0)+(S_1-P_1)(R_0-T_0)}{(R_1-R_0-S_1+S_0-T_1+T_0+P_1-P_0)x^{*}+(S_1-S_0-P_1+P_0)};\\
		\frac{\partial(\pi_C-\pi_D)}{\partial n}|_{P}&=(R_1-R_0-S_1+S_0-T_1+T_0+P_1-P_0)x^{*}+(S_1-S_0-P_1+P_0).\\
	\end{split}
\end{equation}

Due to the restriction on the payoff matrix, $R_0>T_0,S_0>P_0,$ $R_1<T_1, S_1<P_1$,
we have $\frac{\partial(\pi_C-\pi_D)}{\partial n}<0$ in the domain $(0,1)^2$.
As such, the determinant of this Jacobian is positive.
It implies, the two eigenvalues are of the same sign.
Then the stability of the interior fixed point only depends on the trace of $J|_{P}$, i.e., the sum of the two eigenvalues.
Further, the sign of the trace of $J|_{P}$ depends on the sign of $\frac{\partial(\pi_C-\pi_D)}{\partial x}|_P$,
which is equivalent to the sign of $\frac{P_1-S_1}{T_1-R_1}-\frac{S_0-P_0}{R_0-T_0}$.
Thus, the internal fixed point is stable if $\frac{\partial(\pi_C-\pi_D)}{\partial x}|_P$ is negative, otherwise, the matrix (Eq.~\ref{8}) has two positive eigenvalues and the internal fixed point is unstable.
$\hfill\square$

\section{\label{boundary cycle}Stability of heteroclinic cycle}
The state space of system given by Eq.~\ref{4} is $X=[0,1]^2$ with boundary $\partial X$, which is invariant and compact set.
The heteroclinic cycle $\Lambda$ is the $\partial X$ under our dynamics, i.e., $\Lambda=\partial X$.
We leverage the theory of characteristic matrix of the heteroclinic cycle \cite{josef1994}.
The state space $X$ is defined as the intersection of four halfspaces,
\begin{equation}\label{11}
	X=\{x\geq0\}\cap\{1-n\geq0\}\cap\{n\geq0\}\cap\{1-x\geq0\}.
\end{equation}

\noindent Define $x_1=x$, $x_2=1-n$, $x_3=n$, $x_4=1-x$. Then we can rewrite our differential equation (Eq.~\ref{4}) in the form
\begin{equation}\label{12}
	\begin{cases}
		\dot{x_1}=x_1x_4(g(\beta(\pi_C-\pi_D))-g(\beta(\pi_D-\pi_C))) \\
		\dot{x_2}=-\epsilon x_2x_3(-1+(1+\theta)x_1) \\
		\dot{x_3}=\epsilon x_2x_3(-1+(1+\theta)x_1) \\
		\dot{x_4}=-x_1x_4(g(\beta(\pi_C-\pi_D))-g(\beta(\pi_D-\pi_C))).
	\end{cases}
\end{equation}

For convenience, we note
\begin{equation}\label{13}
	\dot{x_j}=x_jf_j(x),   j=1,2,3,4.
\end{equation}

Then the characteristic matrix $C$ is composed by the external eigenvalues
\begin{equation}\label{14}
	C_{k,j}=f_j(F_k),   k,j=1,2,3,4.
\end{equation}
where $F_k$ is the ``boundary'' fixed point.
$C_{k,j}$ denotes the entry in row $k$ and column $j$ in characteristic matrix $C$.
Then the characteristic matrix $C$ reads
\begin{equation}\label{15}
	\begin{array}{c|cccc}
		& x_{1} & x_{2} & x_{3} & x_{4} \\
		\hline(0,0) & g(\beta(S_0-P_0))-g(\beta(P_0-S_0)) & 0 & -\epsilon & 0 \\
		(0,1) & g(\beta(S_1-P_1))-g(\beta(P_1-S_1)) & \epsilon & 0 & 0 \\
		(1,0) & 0 & 0 & \theta\epsilon  &  g(\beta(T_0-R_0))-g(\beta(R_0-T_0)) \\
		(1,1) & 0 & -\theta\epsilon  & 0 & g(\beta(T_1-R_1))-g(\beta(R_1-T_1))
	\end{array}
\end{equation}

\noindent Recall that $S_0-P_0>0$, $S_1-P_2<0$, $T_0-R_0<0$, $T_1-R_1>0$ and the imitation function $g$ is strictly increasing.
We find that each row and each column of $C$ contains only one positive entry and one negative entry, so $\Lambda$ is a simple heteroclinic cycle.

From the corollary 1 in \cite{josef1994}, we know that

\noindent $(i)$ The heteroclinic cycle $\Lambda$ is repelling if C is an M-matrix (all leading principal minors of C are positive).

\noindent $(ii)$ The heteroclinic cycle $\Lambda$ is asymptotically stable if $C$ is not an M-matrix.

From Eq.~\ref{15}, we get that the first, second, and the third leading principle minors of the characteristic matrix $C$ are all positive.
Then whether the characteristic matrix $C$ is an M-matrix depends on the sign of the determinant of $C$, which reads
\begin{equation}\label{16}
	\begin{split}
		det\,C=&\theta\epsilon^2([g(\beta(S_0-P_0))-g(\beta(P_0-S_0))]\times[g(\beta(T_1-R_1))-g(\beta(R_1-T_1))]\\
		&-[g(\beta(S_1-P_1))-g(\beta(P_1-S_1))]\times[g(\beta(T_0-R_0))-g(\beta(R_0-T_0))]).
	\end{split}
\end{equation}
Let $det\,C\neq0$, then $C$ is an M-matrix and $\Lambda$ is repelling if the Inequality.~\ref{9} holds.
And $C$ is not an M-matrix and $\Lambda$ is asymptotically stable if the Inequality.~\ref{10} holds.
$\hfill\square$

\section{\label{Existence of periodic orbit}Existence of periodic orbit (proof of Theorem.~\ref{Theorem.2})}
\begin{proof}
The ODE given by Eq.~\ref{4} is defined on $[0,1]^2$, which is a bounded set,
thus any  $y\in int [0,1]^2$, the $\omega(y)$ is nonempty.
Since Inequality.~\ref{6} holds, based on Theorem.~\ref{Theorem.1}, the internal equilibrium $x^*$ is unique and unstable.
Then there exists a $x\in int [0,1]^2$ such that $x^*\notin \omega(x)$.
On the other hand, $bd[0,1]^2$ is a repelling heteroclinic cycle because Inequality.~\ref{9} holds (See Appendix.~\ref{boundary cycle} for the more detailed proof).
We have $\omega(x)\cap bd[0,1]^2=\emptyset$.
Therefore, $\omega(x)\subseteq int [0,1]^2$.
Based on Lemma.~\ref{lem:01}, we have that
$\omega(x)$ is a closed bounded set in $R^2$,
thus $\omega(x)$ is compact set.
In addition, since $x^*$ is the unique equilibrium by assumption,
it implies that $\omega(x)$ contains no rest point.
By Lemma.~\ref{lem:03},
we have that $\omega(x)$ has to be a periodic orbit.
\end{proof}

\section{\label{Stable limit cycles}Stable limit cycles (proof of Corollary.~\ref{cor:1})}
\begin{proof}
	In the neighborhood of the internal fixed point $x^*$, there exists a $x_1 \in O(x^*) \cap [0,1]^2$, such that $\omega(x_1)$ is a periodic orbit. 
	Then we can find a $x_2$ near $bd[0,1]^2$ satisfying that $\omega(x_2)$ is also a periodic orbit, and the region enclosed by the closed orbit $\omega(x_2)$ includes the region enclosed by the closed orbit $\omega(x_1)$. 
	Consider the region $D$ enclosed by these two closed orbits. 
	Since the internal equilibrium is unstable and the boundary is a repelling heteroclinic cycle, the trajectory outside the region enters the region $D$ and does not go out of the region $D$. 
	By Definition.~\ref{def:3}, there must be a limit cycle in the loop $D$.
\end{proof}

\section{\label{Unstable limit cycles}Unstable limit cycles (proof of Corollary.~\ref{cor:2})}
\begin{proof}
	When the system given by Eq.~\ref{4} satisfies the Inequalities.~\ref{5},~\ref{10}, the internal fixed point is stable and the heteroclinic cycle is asymptotically stable. 
	Based on Lemma.~\ref{lem:04}, the $\alpha$-limit set $\alpha(x)$ by replacing $t \to +\infty$ with $t \to -\infty$ is a periodic orbit similar to the proof of Theorem.~\ref{Theorem.2}.
	We also find an isolated region $D$ that the trajectory inside the loop leaves the region $D$ and does not go in region $D$. 
	Thus, there is an unstable limit cycle. 
\end{proof}

\section{\label{average}Average frequency of stable limit cycles}
The cooperators and the environment are oscillate periodically if the limit cycles are stable.
The average cooperation level and the average environment state are given by,
\begin{equation}
\frac{1}{T} \int_{a}^{a+T} x(t)\, dt = \bar{x}, \quad \frac{1}{T} \int_{a}^{a+T} n(t)\, dt = \bar{n}, 
\end{equation}
where $T$ represents the period of limit cycles.
We find that
\begin{equation}
\frac{d}{dt}(log(n) - log(1-n)) = \frac{\dot{n}}{n} + \frac{\dot{n}}{1-n} = -1 +(1+\theta)x
\end{equation}
It follows by integration that
\begin{equation}
\int_{a}^{a+T} \frac{d}{dt}(log(n(t)) - log(1-n(t)))\,dt = \int_{a}^{a+T}(-1 +(1+\theta)x(t))\, dt,
\end{equation}
that is,
\begin{gather}
(log(n(a+T)) - log(1-n(a+T))) - (log(n(a)) - log(1-n(a))) \notag \\ 
= - T + (1+ \theta) \int_{a}^{a+T}x(t)\, dt.
\end{gather}
Since $n(a)=n(a+T)$, the average cooperation level is 
\begin{equation}
\bar{x} = \frac{1}{T} \int_{a}^{a+T}x(t)\, dt = \frac{1}{1+\theta}.
\end{equation}
Similarly, we have
\begin{equation}
\frac{d}{dt}(log(x) - log(1-x)) = \frac{\dot{x}}{x} + \frac{\dot{x}}{1-x} = g(\beta(\pi_C-\pi_D))-g(\beta(\pi_D-\pi_C))
\end{equation}
It is known that the left term of equal sign is equal to zero. 
To make the equation hold, it satisfies
\begin{equation}
\pi_C - \pi_D = 0,
\end{equation}
since $g$ is a strictly increasing and continuously differentiable function.
Then we get
\begin{equation}
\int_{a}^{a+T}[\pi_C(x(t), n(t)) - \pi_D(x(t), n(t))]\,dt = 0.
\end{equation}
Thus, the average environment state is $n^*=\frac{(T_0-R_0)+\theta(P_0-S_0)}{(R_1-T_1+T_0-R_0)+\theta(S_1-P_1+P_0-S_0)}$.

\bibliographystyle{elsarticle-num}
\bibliography{newtext}

\end{document}